\documentclass[11pt,reqno,notitlepage]{article}
\usepackage{setspace}
\usepackage{amssymb}
\usepackage{amsthm}
\usepackage{graphicx}
\usepackage{amsmath}
\usepackage{fancyhdr}
\usepackage{datetime}
\usepackage[toc]{appendix}
\usepackage[usenames,dvipsnames]{color}
\usepackage{tikz}
\usepackage{color}
\usepackage{colortbl}
\usepackage{mathrsfs}
\usepackage[normalem]{ulem}
\usepackage{graphicx, xcolor}
\usepackage{hyperref}
\usepackage{enumerate}

\usepackage{setspace}
\usepackage{graphicx, xcolor}												
\usepackage[mathscr]{eucal}												
\usepackage{subcaption}
\usepackage{color}
\usepackage{epstopdf}
\usepackage{ifthen}

\allowdisplaybreaks

\makeatletter
\newcommand{\labitem}[2]{%
\def\@itemlabel{\textbf{#1}}
\item
\def\@currentlabel{#1}\label{#2}}
\makeatother

\usetikzlibrary{arrows,decorations.pathreplacing}
\usepackage[top=1in, bottom=1in, left=1in, right=1in]{geometry}
\usepackage{mathtools}									%
\mathtoolsset{showonlyrefs=true}							%

\newtheorem{theorem}{Theorem}
\newtheorem{lemma}[theorem]{Lemma}								%
	
\newtheorem{corollary}[theorem]{Corollary}	
\newtheorem{assumption}[theorem]{Assumption}	
\newtheorem{definition}[theorem]{Definition}
\newtheorem{remark}{Remark}

\numberwithin{equation}{section}	
\numberwithin{theorem}{section}
\def\E{\mathbb{E}}		
		
\def\P{\mathbb{P}}

\def\R{\mathbb{R}}

\def\fcal{\mathcal{F}}
\def\N{\mathbb{N}}

\newcommand{\ind}[1]{\mathbb{I}_{\{#1\}}}

\newcommand{\diag}{\operatorname{diag}}

\def\gamstr{\hat\gamma^{*}}
\def\sigstr{\sigma^I} 
\def\sigM{\sigma^M} 

\newcommand{\dt}{\Delta t}

\newtheorem{example}[theorem]{Example}

\long\def\symbolfootnote[#1]#2{\begingroup\def\thefootnote{\fnsymbol{footnote}}\footnote[#1]{#2}\endgroup}

\newboolean{toggleflag}
\setboolean{toggleflag}{true} 

\begin{document}

\title{A Derivative Pricing Perspective on Liquidity Tokens in Constant Product Market Makers}

\author{
Maxim Bichuch
\thanks{
Department of Mathematics,
SUNY at Buffalo
Buffalo, NY 14260. 
{\tt mbichuch@buffalo.edu}. Work is partially supported by Stellar Development Foundation
Academic Research Grants program. Work  is partially supported by NSF grant DMS- 2420974.}
\and  Zachary Feinstein
\thanks{
School of Business,
Stevens Institute of Technology,
Hoboken, NJ 07030, USA,
{\tt  zfeinste@stevens.edu}. Work is partially supported by Stellar Development Foundation
Academic Research Grants program.}
}
\date{\today}
\maketitle

\begin{abstract}
In decentralized finance, any individual can pool their assets into an automated market maker (AMM) -- herein we focus on the constant product market maker (CPMM) -- in exchange for a claim on  a fraction of future pool assets and fees earned from the market making operations. This position is represented by a liquidity token, whose prevailing on-chain price is effectively the initial deposited assets. Though this price is well-defined, we treat the liquidity token as a derivative position in the prices of the underlying assets for the CPMM in order to deduce risk-neutral pricing and hedging formulas, not dissimilar to the Black-Scholes result. 
Adopting this perspective, in a frictionless environment, hedging the CPMM liquidity token under fair valuation should produce a 
riskless process, which therefore grows at the risk-free rate, something that is not seen in empirical case studies under the prevailing price. 
With our novel pricing formula, we construct a method to calibrate a volatility to data which provides an updated (non-market) valuation which is consistent with the (near-continuous) replication strategy out-of-sample. We conclude with a discussion of novel AMM design considerations motivated by this derivative-pricing perspective. 
\end{abstract}

\section{Introduction}\label{sec:intro}
Decentralized finance (DeFi) is a novel paradigm which seeks to replace financial intermediaries with smart contracts on the blockchain. 
These contracts have been written to act as, e.g., lending platforms, financial exchanges, and insurance providers.
One of the key innovations of the DeFi approach is that these contracts permit investors to add their own liquidity to the ``intermediary'' for a fraction of the fees collected. These investors are often referred to as liquidity providers (LPs) due to the role they take within the financial system.
Within this work, we entirely focus on automated market makers (AMMs) -- in particular the constant product market maker (CPMM)~\cite{xu2021sok,cartea2022decentralised} of Uniswap v1 and v2~\cite{uniswapv2} (and which can readily be constructed in Uniswap v3~\cite{uniswapv3}) -- which construct decentralized exchanges.


We take the view that the investment of a LP in an AMM (i.e., purchasing a liquidity token) is a path-dependent perpetual derivative which earns a dividend stream (i.e., fees) based on the executed swaps on the AMM and its market price is the cost of the liquidity deposited. 
Immediately, with this viewpoint of the liquidity token, the prevailing pricing structure quoted in the CPMM smart contract reveals itself to be based in the pre-Black-Scholes world. 
For example, when applied to Uniswap data, arbitrage opportunities {arise under self-financing hedging strategies} (see Example~\ref{ex:motivating} below).\footnote{Throughout this work, we consider arbitrage only
in a frictionless environment.}\footnote{In contrast to, e.g.,~\cite{fritsch2024measuring} we consider the accounting profits/losses of the liquidity position rather than in relation to an opportunity cost (the loss-versus-rebalancing in the cited paper).} 
Interestingly, even with the potential for arbitrage profits, when compared to the buy-and-hold strategy of the initial liquidity position, nearly 50\% of LPs lose money on Uniswap~\cite{loesch2021impermanent}. 
As LPs form the backbone of these decentralized exchanges, these losses emphasize the need to introduce hedging strategies for liquidity tokens; if investors were to withdraw liquidity \emph{en masse} (due to the high risk of the investment) then the entire DeFi paradigm would fail its primary task to act as a financial intermediary. 
The goal of this work is to construct risk-neutral pricing and hedging theory for liquidity tokens in the CPMM to bring this DeFi product into the modern financial world.

Similar studies have been undertaken previously that highlight different aspects which we will consider. For example \cite{milionis2022automated,hasbrouck2023economic} recognize that the liquidity tokens payoff behave like that of a call option, with the former also applying the idea of hedging and rebalancing for AMMs in continuous time. The approach of those works was extended to discrete time in \cite{milionis2024automated} and further generalized in \cite{nezlobin2025loss}.
However, as far as we are aware, our approach of expressing this optionality characteristic via the volatility of the price process and defining the implied volatility so as to match the price is wholly novel in the literature. 
For instance, \cite{fukasawa2023weighted} presents an approximating static hedge for impermanent loss in the CPMM using variance and gamma swaps. This approach is extended in \cite{lipton2024unified} to consider the concentrated liquidity framework of Uniswap v3. 
In contrast, \cite{milionis2022automated} introduces the loss-versus-rebalancing which considers the cost of dynamic replication of the underlying pool of assets for generic AMMs; \cite{bichuch2025price} uses the loss-versus-rebalancing to quote a risk-neutral expected fees and deduces implications of that construction. Similarly, \cite{cartea2023predictable} introduces a decomposition of this dynamic replication of the underlying pool of assets which allows the introduction of a so-called ``predictable loss''.
Finally, \cite{dewey2023pricing} considers an empirical pricing of a liquidity token in an AMM using historically calibrated parameters.
Herein, rather than concentrating on the underlying holdings of the CPMM, we will consider the stream of fees as the primary driver of the value of the liquidity token. 
We note that \cite{kuan2022liquidity} considers a different approach to estimate the expected fees collected by a LP in a CPMM.

Our primary contributions and innovations for the pricing and hedging of the liquidity position in a CPMM are threefold. 
First, in treating this liquidity position as a derivative on the underlying assets, we find the optimal execution of the position. That is, we deduce conditions for when a risk-neutral investor would optimally invest in (or withdraw from) the CPMM as a LP.
With this optimal execution, second, we are able to produce a risk-neutral valuation for a liquidity token. As a direct consequence, the Greeks and hedging strategies for this position can be readily constructed. As far as the authors are aware, a formal discussion of the Greeks of the liquidity token has never been undertaken previously. Notably, as nearly 50\% of LPs lose money on Uniswap~\cite{loesch2021impermanent}, the introduction of a hedging strategy is of vital importance.
Third, we bring our pricing and hedging theory to data in order to understand its performance in practice. We find that the prevailing market price for the CPMM liquidity token is inconsistent with the cost of the replicating portfolio in a frictionless environment with frequent rebalancing. 
Bringing the theory to the data, we construct a calibrated arbitrage-free price for the liquidity token. 

The rest of this paper is organized as follows. 
In Section~\ref{sec:motivation}, we provide a brief introduction to the mathematics of AMMs and, applying this construction to Uniswap data, we explore the pricing of a CPMM liquidity token in practice. Notably, within Example~\ref{ex:motivating}, we find that the studied data readily admits arbitrage opportunities.
With this motivation, in Section~\ref{sec:cpmm}, we provide the main mathematical theory for risk-neutral pricing and hedging the liquidity token of a CPMM. 
In Section~\ref{sec:repricing}, we provide mathematical discussions surrounding the market implied volatility and a way to calibrate a fair valuation of the liquidity token that does not readily admit arbitrage opportunities. In doing so, we revisit Example~\ref{ex:motivating} and apply the risk-neutral pricing theory to the data so as to calibrate an arbitrage-free pricing of the Uniswap liquidity token. 
Finally, in Section~\ref{sec:discussion}, we provide novel AMM designs which would eliminate these arbitrage opportunities. In proposing these new constructions, we emphasize potential drawbacks which could occur if implemented in practice.

\section{Motivating Example}\label{sec:motivation}

The primary motivation of this paper is to understand how to hedge an investment into the constant product market maker (CPMM). Within Section~\ref{sec:background}, the basic construction of an automated market maker (AMM) -- and specifically a CPMM -- is provided. With these details, in Section~\ref{sec:motivating-example}, we consider the value of the CPMM and its delta hedging position when being priced at the current market rate. This valuation is provided using Uniswap data to demonstrate that arbitrage opportunities exist in the current market setup. The subsequent sections of this work focus on updating the formulas for pricing and hedging so as to properly eliminate arbitrage opportunities.

\subsection{Background on Automated Market Maker}\label{sec:background}

An AMM is, in brief, a pool of assets against which any individual trader can transact. The key innovation of these types of asset pools within decentralized finance is that they permit investors to add their own assets to the AMM in exchange for a fraction of the fees collected by the AMM.
The most common AMM construction is that of the \emph{constant function market maker} (CFMM) which is defined by a multivariate utility function $u: \R^n_+ \to \R_+$ and the size of the asset pool $\Pi \in \R^n_+$ (such that $u(\Pi) > 0$). As summarized in, e.g.,~\cite{angeris2020improved}, the CFMM then permits trades $\delta \in \R^n$ that do not decrease the CFMM's utility nor does it require more assets than the AMM holds, i.e., $\delta$ is a valid trade if:
\[u(\Pi) \leq u(\Pi + \delta) \quad \text{ and } \quad \Pi + \delta \in \R^n_+.\]
These AMMs are called constant function market makers because, under mild assumptions (see, e.g., \cite{bichuch2022axioms}), the utility before and after a transaction are equal ($u(\Pi) = u(\Pi + \delta)$).

Based on the constant function construction, the marginal price of asset $i$ in terms of the num\'eraire asset $j$ can be determined via the relation $P_i^j(\Pi) = \frac{\partial}{\partial \Pi_i} u(\Pi) / \frac{\partial}{\partial \Pi_j} u(\Pi)$. This mapping is often called the pricing oracle.
As is clear from the construction, $P_i^j(\Pi) = P_j^i(\Pi)^{-1}$ for any $\Pi \in \R^n_+$; this makes clear that there do not exist fees within this construction.
In practice, fees are charged on a fraction of the assets being sold to the AMM, i.e., $\gamma \in (0,1)$ of the incoming assets (asset $i$ such that $\delta_i > 0$) are taken to compensate the market maker for its service as a counterparty. Mathematically this modifies the CFMM construction so that $u(\Pi) = u(\Pi + [I - \gamma\diag(\ind{\delta > 0})]\delta)$. 

Finally, investors are able to deposit their assets into the AMM in exchange for a fraction of the fees collected; these depositors are often referred to as liquidity providers (LPs). This is done following the constant pricing oracle construction so that the prices before and after the liquidity provision remain equal, i.e., $\delta \in \R^n_+$ is deposited if $P_i^j(\Pi) = P_i^j(\Pi + \delta)$ for any pair of assets $(i,j)$; in practice, and as presented explicitly in the preceding equation, no fees are charged on these deposits.
Similarly, LPs can later withdraw their assets at any time taking from the AMM the same fraction of the asset pool that they initially deposited.
The fraction of fees collected by any individual LP is equal to the fraction of the assets that they hold at the AMM. Throughout this work, and as implemented within Uniswap v3 \cite{uniswapv3}, the fees are distributed immediately upon collection to the liquidity providers.

\begin{assumption}\label{ass:2asset}
For the remainder of this work, we will consider the CPMM 
in the $n = 2$ asset setting as is utilized by Uniswap and Sushiswap pools (i.e., $u(x,y) := xy$). In addition, throughout, we take the second asset as the num\'eraire (i.e., $P(x,y) := P_1^2(x,y)$).
\end{assumption}

In practice, AMMs exist as smart contracts that operate in a decentralized manner directly on a blockchain. This means that transactions are only processed at the block-writing times. By construction of the blockchain, this occurs at discrete times. For the Bitcoin blockchain, which utilizes a proof-of-work construction, each block is processed in approximately 10 minute intervals. However, more modern blockchains -- using, e.g., proof-of-stake consensus -- have nearly constant inter-block times $\dt > 0$. For the Ethereum blockchain, the inter-block time is $\dt = 12$ seconds; for the Polygon blockchain, the inter-block time is $\dt = 2$ seconds.

\begin{assumption}\label{ass:block}
For the remainder of this work, we will assume that the inter-block time is fixed at $\dt > 0$. Therefore the realized price process is $P_{i\dt}$ for $i \in \N$.
\end{assumption}

\subsection{Hedging a CPMM}\label{sec:motivating-example}
Consider the CPMM $u(x,y) = xy$. By construction, its pricing oracle $P(x,y) = y/x$ is given by the ratio of the pool's asset holdings. Due to this construction, LPs mint new liquidity tokens (i.e., deposit) by providing assets at the same ratio as the pool currently holds. Uniswap v2~\cite{uniswapv2} defines the number of outstanding liquidity tokens to be $L = \sqrt{xy}$ where $(x,y)$ is are the shares of assets held by the pool.
Notably, for the CPMM, there is a one-to-one relation between the asset holdings of the pool $(x,y)$ and, jointly, the price $P$ and the number of liquidity tokens $L$. Already we have provided how $P,L$ are constructed from the asset holdings; conversely, given the price and liquidity tokens: $x = L / \sqrt{P}$ and $y = L \sqrt{P}$.
In fact, these relations make clear that the value held in the pool is given by $Px + y = 2L \sqrt{P}$.
\begin{remark}
Because of the linear relation between the value of the pool and the value of a liquidity position (i.e., scaled by $L$), throughout the remainder of this work we consider the value of a single liquidity token $L = 1$ except where otherwise explicitly stated.
\end{remark}

Given this value of the liquidity token, we can eliminate entirely the riskiness of the position by trading the underlying token appropriately (see, e.g., the discussion in and preceding \cite[Chapter 2.4]{karatzas1998methods}). Specifically, we can hedge this position by holding $\Delta$ units of the underlying where $\Delta$ is the sensitivity of the position to the token price, i.e., $\Delta = \frac{d}{dP} 2\sqrt{P} = 1/\sqrt{P}$.
In theory, by purchasing and rebalancing this position, we expect to perfectly hedge the liquidity token; however, this hedge may fail during periods of extreme volatility as delta hedging only provides protections along the linearly approximated portfolio. 
In the following numerical example, we take data on a Uniswap pool in order to look for any arbitrage opportunities, i.e., for trends in the hedged position. 
Surprisingly, given the theory on delta hedging, we observe a persistent positive drift in the discounted, frequently-hedged position. 
To preview future results, in the next section we provide a theory of risk-neutral pricing of a liquidity token and deduce the Greeks of this liquidity position; these results are summarized in Table~\ref{tab:greeks}.
Herein, we consider the hedging and rebalancing strategy under negligible fees as is taken in the Black-Scholes framework.

\begin{example}\label{ex:motivating}
Consider a USDC/WETH Uniswap v3 pool on the Polygon blockchain ($\dt = 2$ seconds) with $\gamma = 5bps$ fees between market creation (December 20, 2021) and June 21, 2025 with the investment made over the entire price line to mimic the CPMM.\footnote{This data is taken from smart contract 0x45dda9cb7c25131df268515131f647d726f50608.} 
This pool was chosen as it is a representative Uniswap v3 pool with high liquidity.
For this example, we will assume that the risk-free rate $r = 0\%$  throughout the period of study for simplicity and to match the rate earned on USDC.
In Figure~\ref{fig:motivating-full}, we show the discounted values of both a 1 USDC liquidity position (i.e., $\sqrt{P}/\sqrt{P_0}$ plus the collected fees) and of the delta hedged positions (with rebalancing of the delta hedge at different frequencies ranging from every block to every day). As expected, the volatility of the hedged position is significantly lower than that of the unhedged position. However, as made clear in Figure~\ref{fig:motivating-zoom}, when delta is hedged every block, then the hedged position has a distinct \emph{positive} trend line indicating the presence of an arbitrage opportunity except on distinct dates in which the hedging fails due to large intra-block price swings. We hypothesize that a better hedging strategy (e.g., Gamma hedging) is needed to reduce these jumps. Notably, from February 2022 to March 2024, there is a nearly unbroken stretch of profitable arbitrage opportunities by implementing this delta hedging strategy. As we reduce the hedging frequency, we find a positive trend exists for rehedging every minute as well. 
\begin{figure}[h!]
\centering
\begin{subfigure}[t]{0.45\textwidth}
\centering
\includegraphics[width=\textwidth]{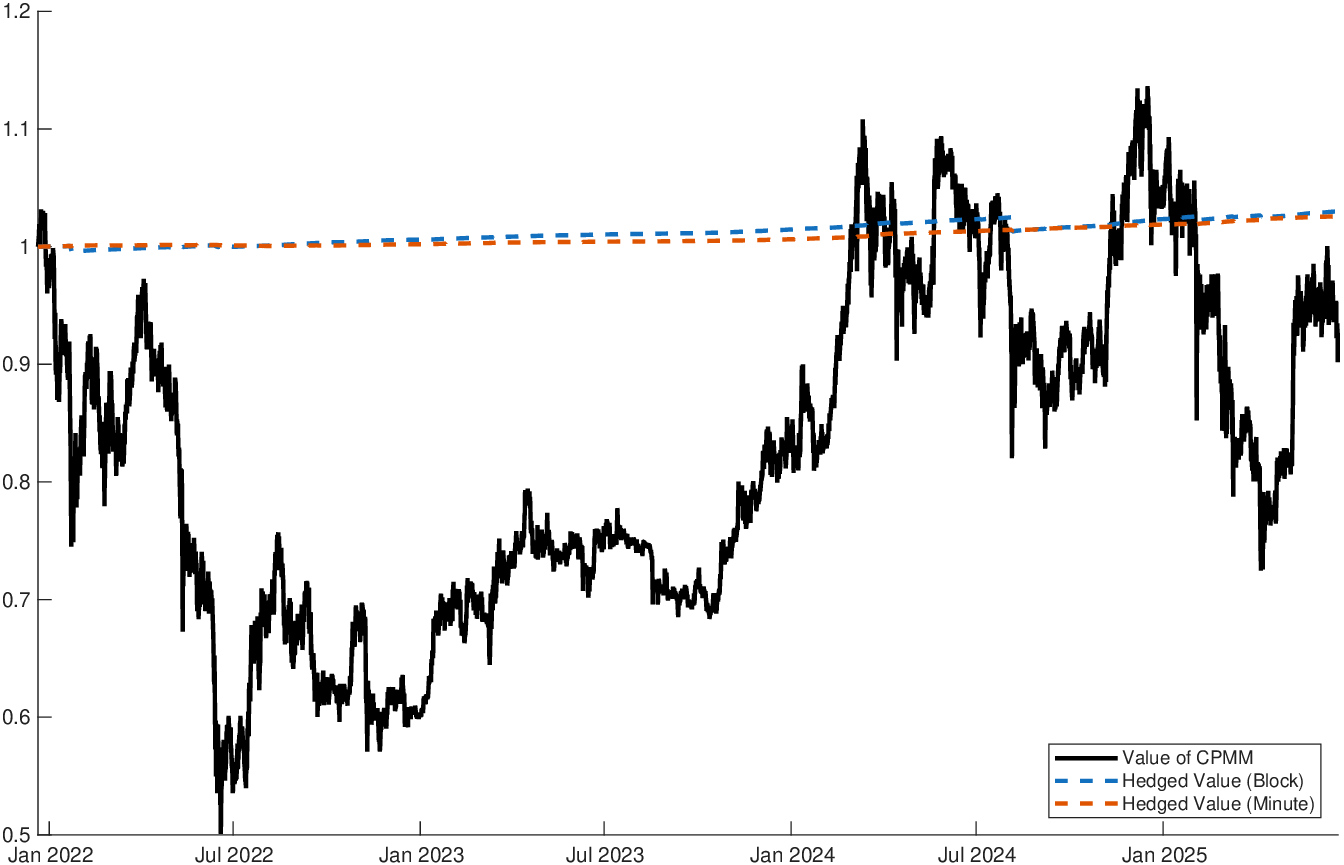}
\caption{Discounted values over time.}
\label{fig:motivating-full}
\end{subfigure}
~
\begin{subfigure}[t]{0.45\textwidth}
\centering
\includegraphics[width=\textwidth]{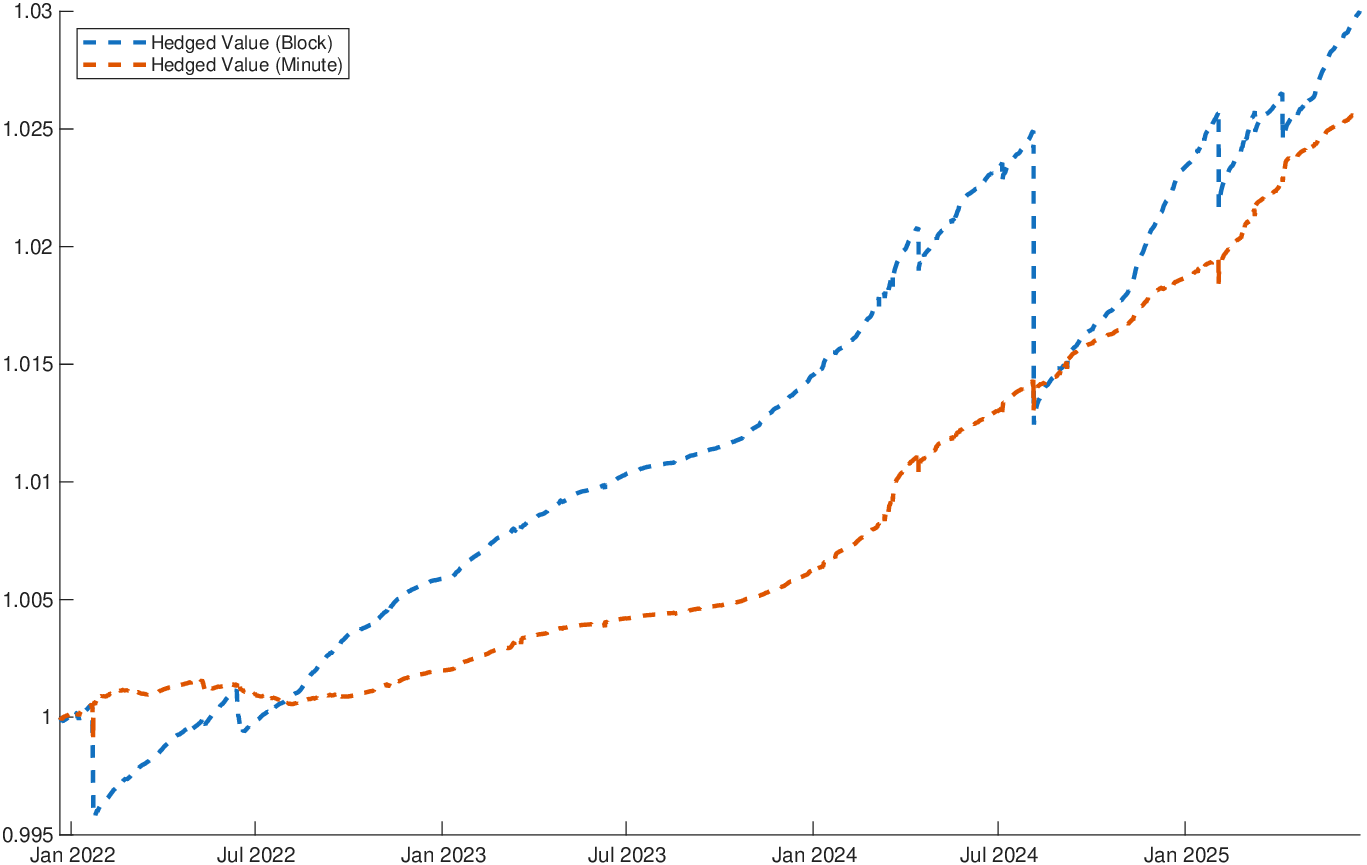}
\caption{Zoomed in image of the discounted delta hedged position.}
\label{fig:motivating-zoom}
\end{subfigure}
\caption{Example~\ref{ex:motivating}: Comparison of the discounted values of a 1 USDC investment in liquidity tokens and its delta hedged position from market creation (December 20, 2021) until June 21, 2025.}
\label{fig:motivating}
\end{figure}
Before continuing, we wish to expand on the sensitivity of this arbitrage to the chosen risk-free rate $r = 0\%$. (All parameters besides the risk-free rate $r$ are specified by the Polygon blockchain or Uniswap contract.) Notably, as the risk-free rate increases, the profitability of this delta hedging strategy decreases due to the time-value of money that is being forfeited by holding USDC. This leads us to conjecture that the profitability would return if a tokenized money market account was used in place of USDC. 
\end{example}

\section{Constant Product Market Making}\label{sec:cpmm}

Within this section, and as discussed in Assumption~\ref{ass:2asset}, we consider the CPMM $u: \R^2_+ \to \R_+$ defined by $u(x,y) = xy$. 
This AMM construction is the most widely reported utility function and has been implemented by, e.g., Uniswap v2 and SushiSwap, and the necessary details are provided in Section~\ref{sec:motivating-example}.
As provided above, we will typically take advantage of the equivalence between the price process $(P_t)$ and the process of asset holdings $(x_t,y_t)$ corresponding to a single liquidity token.
For ease of notation, when given the price process, we define the asset holdings of the pool accordingly $(x_t,y_t) = (1/\sqrt{P_t},\sqrt{P_t})$.

Within this section, first we present the market model for the price process $(P_t)$ in Section~\ref{sec:market}. Then, in Section~\ref{sec:cpmm-value}, we provide a risk-neutral valuation of a single liquidity token for the CPMM. This valuation involves solving an optimal stopping problem to determine when the LP should withdraw her funds from the CPMM pool.
Finally, within Section~\ref{sec:greeks}, we will provide select Greeks for a liquidity token of the CPMM.

\begin{table}[t]
\centering
\begin{tabular}{|l||c|l|}
\hline
\textbf{Valuation} & $V_0(P)$ & $\frac{2\hat\gamma\sqrt{P}}{\gamstr}$ \\ \hline\hline
\textbf{Delta} & $\frac{\partial}{\partial P} V_0(P)$ & $\frac{\hat\gamma}{\gamstr\sqrt{P}}$ \\ \hline
\textbf{Gamma} & $\frac{\partial^2}{\partial P^2} V_0(P)$ & $-\frac{\hat\gamma}{2\gamstr P^{3/2}}$ \\ \hline
\textbf{Vega} & $\frac{\partial}{\partial \sigma} V_0(P)$ & 
$\frac{\hat\gamma \sqrt{P} e^{-\frac{1}{2}(r+\frac{\sigma^2}{4})\dt}}{1 - e^{-\frac{1}{2}(r+\frac{\sigma^2}{4})\dt}} \left[\sqrt{\frac{\dt}{2\pi}}e^{-\frac{r^2\dt}{2\sigma^2}} - \frac{\sigma\dt}{4}\frac{\Phi(\frac{(r+\frac{\sigma^2}{2})\sqrt{\dt}}{\sigma})-e^{-r\dt}\Phi(\frac{(r-\frac{\sigma^2}{2})\sqrt{\dt}}{\sigma})}{1 - e^{-\frac{1}{2}(r+\frac{\sigma^2}{4})\dt}}\right]$
 \\ \hline
\end{tabular}
\caption{Summary table of valuation (Theorem~\ref{thm:v2}) and Greeks (Section~\ref{sec:greeks}) of a single liquidity token in a CPMM when $\hat\gamma \geq \gamstr$ as defined in~\eqref{eq:gamstr}.}
\label{tab:greeks}
\end{table}

\subsection{Market Model}\label{sec:market}

As discussed above in Section~\ref{sec:motivating-example}, the cost of constructing a single liquidity token within a CPMM is $P_0 x_0 + y_0 = 2\sqrt{P_0}$ at price of $P_0 > 0$. Once constructed, and following Assumption~\ref{ass:block}, the liquidity position of the CPMM is, simply, a perpetual Bermudan option. That is, it is a derivative of the price process $(P_t)$ that can be exercised at any block but can continue indefinitely until it is exercised. Until exercise, the AMM disperses fees to the amount of $\frac{\gamma}{1-\gamma}[P_{i\dt}(x_{i\dt}-x_{(i-1)\dt})^+ + (y_{i\dt} - y_{(i-1)\dt})^+]$ at block $i \in \N$, i.e., at time $i\dt$, for fee level $\gamma \in (0,1)$. By construction, exactly one of the terms $(x_{i\dt}-x_{(i-1)\dt})$ and $(y_{i\dt}-y_{(i-1)\dt})$ is positive. Throughout the remainder of this work, we take $\hat\gamma := \frac{\gamma}{1-\gamma}$ to be the fraction of the change in the pool reserves collected by the LPs. 

\begin{assumption}\label{ass:gbm}
For the remainder of this work, following \cite{cao2025structural}, we will assume that the price process $(P_t)$ follows the risk-neutral geometric Brownian motion
\begin{align}
dP_t = P_t [r dt + \sigma dW_t]
\label{eq:dP}
\end{align}
for risk-free rate $r \geq 0$, volatility $\sigma > 0$, Brownian motion $W$ on a filtered probability space $(\Omega,\fcal,(\fcal_t^W)_{t\ge0},\P)$, with the filtration $(\fcal_t^W)_{t\ge0}$ generated by the Brownian motion,
and initial value $P_0$. 
\end{assumption}
\begin{remark}
As highlighted within Assumption~\ref{ass:gbm}, the price process considered herein is consistent with that in~\cite{cao2025structural}. Specifically, that work considers a terminal, log-normally distributed price and calculate the resulting equilibrium from a system of optimizing agents (traders and LPs). While \cite{cao2025structural} considers a static setting, we extend to the dynamic setting following the same general direction: we assume traders actualize this log-normal distribution in the AMM.
\end{remark}
%
\begin{remark}
Following Assumption~\ref{ass:gbm}, the measure $\P$ is a risk-neutral measure. In fact, due to the completeness of the constructed market, $\P$ is the unique risk-neutral measure.\footnote{We refer the interested reader to \cite{harrison1981martingales,delbaen1994general} for the discussion on the Fundamental Theorems of Asset Pricing and to \cite[Chapters 1 and 2]{karatzas1998methods} that show the existence of a hedge in a complete market driven by a Brownian Motion as well as pricing by replication in such a market.}
Throughout the remainder of this work, we will exploit this fact in order to consider the pricing and hedging of a liquidity token in a CPMM.
\end{remark}
\begin{remark}
Implicitly, when constructing the fees, we are assuming there is only one transaction in each block to align the price of the CPMM to $P$. When studying the valuation of the CPMM liquidity token in Section~\ref{sec:cpmm-value} below, this assumption guarantees that we find a \emph{lower bound} on the risk-neutral value as other (uninformed) trades may also occur; such trades increase the fees collected by the LPs without altering the fundamental price process. 
\end{remark}
\begin{remark}
We assume that any strategy that holds the num\'eraire asset (the second asset) instantaneously deposits it into the money market account so as to earn the risk-free rate $r$. We stress that this applies only to strategies as opposed to the liquidity provided to the CPMM pool over which the LP has no control.
\end{remark}

\subsection{Risk-Neutral Valuation}\label{sec:cpmm-value}
Within this section, our goal is to quantify the risk-neutral price for a single liquidity token of the CPMM. Due to the perpetual Bermudan option construction of the liquidity position, the value of the LP position is the maximum of either withdrawing at that price (i.e., $2\sqrt{P}$) or the discounted expectation of continuing for another block. Mathematically, under Assumption~\ref{ass:gbm}, this is provided by the value function $V_0: \R_{++} \to \R_{++}$ defined by
\begin{align}
V_0(P) &= \max\{2\sqrt{P} \; , \; e^{-r\dt}\E[V_0(Pe^{(r-\frac{\sigma^2}{2})\dt + \sigma B_{\dt}}) + \hat\gamma F(P,Pe^{(r-\frac{\sigma^2}{2})\dt + \sigma B_{\dt}})]\}, \label{eq:V0}\\
F(P_0,P_1) &:= P_1\left(\frac{1}{\sqrt{P_1}} - \frac{1}{\sqrt{P_0}}\right)^+ + \left(\sqrt{P_1} - \sqrt{P_0}\right)^+.\label{eq:F}
\end{align}
In other words, the value function $V_0$ in \eqref{eq:V0} is comprised of the stopping value (withdrawing $2\sqrt{P}$), the first term, or the continuation value (the discounted value at the next block, which by our assumption is the value function at the next block's price, together with the proportion $\hat \gamma$ of the fees $F$ collected), the second term.
\begin{remark}
Because of the proportionality rule for the disbursement of fees, the value of an arbitrary number $L$ of liquidity tokens is equal to $L V_0(P)$ at the current market price of $P > 0$.
\end{remark}

Immediately, we are able to construct the risk-neutral price for a liquidity token in a CPMM. 
\begin{theorem}\label{thm:v2}
Fix the risk-free rate $r \geq 0$ and let the price process follow the geometric Brownian motion as in Assumption~\ref{ass:gbm}. 
Assume the current time ($t = 0$) is a block time.
A risk-neutral investor will deposit liquidity in the constant product market maker if, and only if,  
\begin{align}
\hat\gamma \geq \gamstr := 2\left[-1 + \frac{\Phi\left(\frac{(r+\frac{\sigma^2}{2})\sqrt{\dt}}{\sigma}\right) - e^{-r\dt}\Phi\left(\frac{(r-\frac{\sigma^2}{2})\sqrt{\dt}}{\sigma}\right)}{1 - e^{-\frac{1}{2}(r + \frac{\sigma^2}{4})\dt}}\right]^{-1},
\label{eq:gamstr}
\end{align}
where $\Phi$ is the CDF of the standard normal distribution.
Provided $\hat\gamma \geq\gamstr$, the value of the liquidity token at the current price $P_0 > 0$ is given by
\begin{equation}
V_0(P_0) = \frac{2\hat\gamma\sqrt{P_0}}{\gamstr}.
\end{equation}
\end{theorem}
\begin{proof}
Throughout this proof, let $Z \sim N(0,1)$ follow the standard normal distribution.
First, consider the expectation of the discounted fees. That is, given initial block price $P_0 > 0$, we want to find:
\begin{align*}
\bar F_0 &= e^{-r\dt}\E[F(P_0,P_{\dt})] = e^{-r\dt}\E[F(P_0,P_0e^{(r-\frac{\sigma^2}{2})\dt + \sigma Z \sqrt{\dt}})] \\
&= \sqrt{P_0}\E\left[e^{-\frac{\sigma^2}{2}\dt + \sigma Z \sqrt{\dt}}\left(e^{-\frac{1}{2}(r-\frac{\sigma^2}{2})\dt - \frac{\sigma}{2}Z\sqrt{\dt}} - 1\right)^+ + e^{-r\dt}\left(e^{\frac{1}{2}(r-\frac{\sigma^2}{2})\dt + \frac{\sigma}{2}Z\sqrt{\dt}}-1\right)^+\right] \\
&= \sqrt{P_0}\E\left[\begin{array}{l} \left(e^{-\frac{1}{2}(r+\frac{\sigma^2}{2})\dt+\frac{\sigma}{2}Z\sqrt{\dt}} - e^{-\frac{\sigma^2}{2}\dt + \sigma Z \sqrt{\dt}}\right)\ind{(r-\frac{\sigma^2}{2})\dt + \sigma Z\sqrt{\dt} < 0} \\ + \left(e^{-\frac{1}{2}(r+\frac{\sigma^2}{2})\dt+\frac{\sigma}{2}Z\sqrt{\dt}}-e^{-r\dt}\right)\ind{(r+\frac{\sigma^2}{2})\dt + \sigma Z \sqrt{\dt} > 0} \end{array}\right] \\
&= \sqrt{P_0}\E[e^{-\frac{1}{2}(r+\frac{\sigma^2}{2})\dt + \frac{\sigma}{2}Z\sqrt{\dt}}] - \sqrt{P}\E\left[e^{-\frac{\sigma^2}{2}\dt + \sigma Z \sqrt{\dt}}\ind{Z < -\frac{(r-\frac{\sigma^2}{2})\sqrt{\dt}}{\sigma}}\right] \\
    &\qquad - \sqrt{P}e^{-r\dt}\E\left[\ind{Z > -\frac{(r-\frac{\sigma^2}{2})\sqrt{\dt}}{\sigma}}\right] \\
&= \sqrt{P_0} e^{-\frac{1}{2}(r + \frac{\sigma^2}{4})\dt} - \sqrt{P}\E\left[\ind{Z+\sigma\sqrt{\dt} < -\frac{(r-\frac{\sigma^2}{2})\sqrt{\dt}}{\sigma}}\right] \\
    &\qquad - \sqrt{P}e^{-r\dt}\left[1-\Phi\left(-\frac{(r-\frac{\sigma^2}{2})\sqrt{\dt}}{\sigma}\right)\right] \\
&= \sqrt{P_0}\left[e^{-\frac{1}{2}(r+\frac{\sigma^2}{4})\dt} - 1 + \Phi\left(\frac{(r+\frac{\sigma^2}{2})\sqrt{\dt}}{\sigma}\right) - e^{-r\dt}\Phi\left(\frac{(r-\frac{\sigma^2}{2})\sqrt{\dt}}{\sigma}\right)\right] \\
&= \frac{2 (1 - e^{-\frac{1}{2}(r+\frac{\sigma^2}{4})\dt}) \sqrt{P_0}}{\gamstr}.
\end{align*}
It similarly follows that $\E[F(P_{i\dt},P_{(i+1)\dt}) \; | \; \fcal_{i\dt}] =  \frac{2 (1 - e^{-\frac{1}{2}(r+\frac{\sigma^2}{4})\dt}) \sqrt{P_{i\dt}}}{\gamstr} = \bar F_0\sqrt{P_{i\dt}/P_0}$ for any block $i$.

Consider the ansatz that an investor would choose to either never invest in the CPMM or, once invested, never withdraw her liquidity from the CPMM. Let $\tilde V(P)$ denote the value of the perpetual fees collection, i.e.,
\begin{align*}
\tilde V_0(P_0) &= \hat\gamma\sum_{i = 0}^\infty e^{-r i \dt} \E[F(P_{i\dt},P_{(i+1)\dt})] \\
&= \hat\gamma\sum_{i = 0}^\infty e^{-r i \dt} \E[\E[F(P_{i\dt},P_{(i+1)\dt}) \; | \; \fcal_{i\dt}]] \\
&= \hat\gamma\bar F_0 \sum_{i = 0}^\infty e^{-r i \dt} \E[\sqrt{P_{i\dt}/P_0}] \\
&= \hat\gamma\bar F_0 \sum_{i = 0}^\infty \E[e^{-\frac{1}{2}(r + \frac{\sigma^2}{2})i\dt + \frac{\sigma}{2}W_{i\dt}}] \\
&= \frac{2 \hat\gamma (1 - e^{-\frac{1}{2}(r+\frac{\sigma^2}{4})\dt}) \sqrt{P_0}}{\gamstr} \sum_{i = 0}^\infty e^{-\frac{1}{2}(r+\frac{\sigma^2}{4})i\dt} \\
&= \frac{2\hat\gamma\sqrt{P_0}}{\gamstr}.
\end{align*}
Therefore, by inspection, this strategy implies that the investor should deposit liquidity into the CPMM if, and only if, $\tilde V_0(P_0) \geq 2\sqrt{P_0}$, i.e., $\hat\gamma \geq\gamstr$.
Thus, we construct the ansatz value function $V_0(P_0) = \tilde V_0(P_0) \ind{\hat\gamma \geq\gamstr} + 2\sqrt{P_0} \ind{\hat\gamma <\gamstr}$ (noting that $\tilde V_0(P_0) = 2\sqrt{P_0}$ at $\hat\gamma =\gamstr$).

It remains to verify that this ansatz strategy is optimal. 
First, assume $\hat\gamma \geq\gamstr$ so that $V_0(P_0) = \tilde V_0(P_0)$.  By construction, we note that $\tilde V_0(P_0)$ coincides with its continuation value, that is, $\tilde V_0(P_0) $ $= e^{-r\dt}\E[\tilde V_0(P_0e^{(r-\frac{\sigma^2}{2})\dt + \sigma Z \sqrt{\dt}}) + \hat\gamma F(P_0,P_0e^{(r-\frac{\sigma^2}{2})\dt + \sigma Z \sqrt{\dt}})]$ for every $P_0 > 0$. 
Indeed, 
\begin{align*}
&e^{-r\dt}\E[\tilde V_0(P_0e^{(r-\frac{\sigma^2}{2})\dt + \sigma Z \sqrt{\dt}}) + \hat\gamma F(P_0,P_0e^{(r-\frac{\sigma^2}{2})\dt + \sigma Z \sqrt{\dt}})]\\
&=e^{-r\dt} \E\Bigg[\frac{2\hat\gamma\sqrt{P_0e^{(r-\frac{\sigma^2}{2})\dt + \sigma Z \sqrt{\dt}} }}{\gamstr}    \Bigg] + \hat\gamma \frac{2 (1 - e^{-\frac{1}{2}(r+\frac{\sigma^2}{4})\dt}) \sqrt{P_0}}{\gamstr} \\
&=  \frac{2\hat\gamma\sqrt{P_0}}{\gamstr} = \tilde V(P) .
\end{align*}
Furthermore, by $\hat\gamma \geq\gamstr$, this continuation value is always at least as large as the stopping (i.e. withdrawing) value $2\sqrt{P_0}$ validating the construction of $V_0(P_0)$.
Second, assume $\hat\gamma <\gamstr$ so that $V_0(P_0) = 2\sqrt{P_0}$. The continuation value under this value function is $2\sqrt{P_0}[e^{-\frac{1}{2}(r + \frac{\sigma^2}{4})\dt} + (1-e^{-\frac{1}{2}(r+\frac{\sigma^2}{4})\dt})\hat\gamma/\gamstr] < 2\sqrt{P_0}$ by assumption.
Therefore, this construction satisfies the dynamic programming principle and the proof is complete.
\end{proof}

We wish to conclude our discussion of the risk-neutral value of a liquidity token by considering its value when not at a block time. Recall from Assumption~\ref{ass:gbm} that the price process is observable continuously in time even though the blockchain only allows transactions at the block times. 
\begin{corollary}\label{cor:v2}
Consider the risk-free rate $r \geq 0$ and let the price process follow the geometric Brownian motion as in Assumption~\ref{ass:gbm}. 
Assume $t \in (0,\dt)$ is an inter-block time and let $\tau := \dt - t$ be the time until the next block.
Set the prior block time price to be $P_0 > 0$.
Provided $\hat\gamma \geq\gamstr$, the value of the liquidity token at the current time $t$ and price $P_t > 0$ is given by
\begin{align*}
V_t(P_t) = &\left(\frac{2}{\gamstr} + 1\right)\hat\gamma e^{-\frac{1}{2}(r+\frac{\sigma^2}{4})\tau}\sqrt{P_t} - \hat\gamma \frac{P_t}{\sqrt{P_0}} \left[1 - \Phi\left(\frac{\log(P_t/P_0)+(r+\frac{\sigma^2}{2})\tau}{\sigma\sqrt{\tau}}\right)\right]\\
    &\qquad - \hat\gamma e^{-r\tau}\sqrt{P_0}\Phi\left(\frac{\log(P_t/P_0)+(r-\frac{\sigma^2}{2})\tau}{\sigma\sqrt{\tau}}\right).
\end{align*}
\end{corollary}
\begin{proof}
As the blockchain only permits transactions at block times, by construction the value of the liquidity token between block times is given by the expected value at the next block, i.e.,
$V_t(P_t) = e^{-r\tau}\E[V_0(P_{\dt}) + \hat\gamma F(P_0,P_{\dt}) \; | \; \fcal_t]$.
As in the proof of Theorem~\ref{thm:v2}, let $Z \sim N(0,1)$ follow the standard normal distribution.
Consider, first, the discounted value of the liquidity token at the next block:
\begin{align*}
e^{-r\tau}\E[V_0(P_{\dt}) \; | \; \fcal_t] &= \frac{2\hat\gamma}{\gamstr}\sqrt{P_t}e^{-r\tau}\E[e^{\frac{1}{2}(r-\frac{\sigma^2}{2})\tau + \frac{\sigma}{2} Z\sqrt{\tau}}] \\
&= \frac{2\hat\gamma}{\gamstr}e^{-\frac{1}{2}(r+\frac{\sigma^2}{4})\tau}\sqrt{P_t}.
\end{align*}
Consider, now, the discounted value of the fees that would be earned in this current block:
\begin{align*}
&e^{-r\tau}\hat\gamma\E[F(P_0,P_{\dt}) \; | \; \fcal_t]\\
&= e^{-r\tau}\hat\gamma\E\left[P_t e^{(r-\frac{\sigma^2}{2})\tau + \sigma Z \sqrt{\tau}} \left(P_t^{-1/2}e^{-\frac{1}{2}(r-\frac{\sigma^2}{2})\tau - \frac{\sigma}{2} Z \sqrt{\tau}} - P_0^{-1/2}\right)^+ \; | \; \fcal_t\right]\\
&\qquad+ e^{-r\tau}\hat\gamma\E\left[\left(P_t^{1/2} e^{\frac{1}{2}(r-\frac{\sigma^2}{2})\tau + \frac{\sigma}{2} Z \sqrt{\tau}} - P_0^{1/2}\right)^+ \; | \; \fcal_t\right] \\
&= e^{-r\tau}\hat\gamma\E\left[\sqrt{P_{\dt}} - \frac{P_t e^{(r+\frac{\sigma^2}{2})\tau + \sigma Z \sqrt{\tau}}}{\sqrt{P_0}}\ind{(r+\frac{\sigma^2}{2})\tau + \sigma Z \sqrt{\tau} < -\log(P_t/P_0)}\; | \; \fcal_t\right]\\
&\qquad- e^{-r\tau}\hat\gamma\E\left[ \sqrt{P_0} \ind{(r+\frac{\sigma^2}{2})\tau + \sigma Z \sqrt{\tau} > -\log(P_t/P_0)} \; | \; \fcal_t\right] \\
&= \hat\gamma e^{-\frac{1}{2}(r+\frac{\sigma^2}{4})\tau}\sqrt{P_t} - \hat\gamma\frac{P_t}{\sqrt{P_0}}\E\left[e^{-\frac{\sigma^2}{2}\tau + \sigma Z \sqrt{\tau}} \ind{Z < -\frac{\log(P_t/P_0) + (r - \frac{\sigma^2}{2})\tau}{\sigma\sqrt{\tau}}}\right]\\
    &\qquad - \hat\gamma e^{-r\tau} \sqrt{P_0} \P\left(Z > -\frac{\log(P_t/P_0) + (r - \frac{\sigma^2}{2})\tau}{\sigma\sqrt{\tau}}\right) \\
&= \hat\gamma e^{-\frac{1}{2}(r+\frac{\sigma^2}{4})\tau}\sqrt{P_t} - \hat\gamma\frac{P_t}{\sqrt{P_0}}\P\left(Z + \sigma\sqrt{\tau} < -\frac{\log(P_t/P_0) + (r - \frac{\sigma^2}{2})\tau}{\sigma\sqrt{\tau}}\right)\\
    &\qquad - \hat\gamma e^{-r\tau} \sqrt{P_0} \Phi\left(\frac{\log(P_t/P_0) + (r - \frac{\sigma^2}{2})\tau}{\sigma\sqrt{\tau}}\right) \\
&= \hat\gamma e^{-\frac{1}{2}(r+\frac{\sigma^2}{4})\tau}\sqrt{P_t} - \hat\gamma\frac{P_t}{\sqrt{P_0}}\Phi\left(-\frac{\log(P_t/P_0) + (r + \frac{\sigma^2}{2})\tau}{\sigma\sqrt{\tau}}\right) \\
&\qquad- \hat\gamma e^{-r\tau} \sqrt{P_0} \Phi\left(\frac{\log(P_t/P_0) + (r - \frac{\sigma^2}{2})\tau}{\sigma\sqrt{\tau}}\right).
\end{align*}
Combining these terms together immediately provides the desired result.
\end{proof}

\subsection{Greeks}\label{sec:greeks}
As we can describe the value of the liquidity position in the CPMM via Theorem~\ref{thm:v2} and Corollary~\ref{cor:v2}, it is valuable also to understand how to hedge the risks of this position. For this purpose we will consider various Greeks for the liquidity token. Herein we will focus specifically on the Greeks at block times though, utilizing the forms of Corollary~\ref{cor:v2}, the Greeks can also be computed between blocks.
\begin{assumption}\label{ass:gamma}
Following Theorem~\ref{thm:v2}, throughout this section, we will assume that $\gamstr \leq \hat\gamma$.
\end{assumption}

\paragraph{Delta:} First, consider the sensitivity of the value of the liquidity token to the underlying price, i.e., the delta of the liquidity token.  By construction in Theorem~\ref{thm:v2}, this sensitivity is driven entirely by the square root of the current price, i.e., 
\begin{align}
\frac{\partial}{\partial P} V_0(P) = \frac{\hat\gamma}{\gamstr\sqrt{P}} = \frac{V_0(P)}{2P}
\label{eq:Delta}
\end{align}
for any price $P > 0$.
Notably, by this construction, it immediately follows that the delta $\frac{\partial}{\partial P} V_0(P) > 0$ is strictly positive.

\paragraph{Gamma:} Second, consider the sensitivity of the delta of the liquidity token to the underlying price, i.e., the gamma of the liquidity token. Much like the delta above, this sensitivity follows simply from Theorem~\ref{thm:v2}: 
\begin{align}
\frac{\partial^2}{\partial P^2} V_0(P) = -\frac{\hat\gamma}{2\gamstr P^{3/2}} = -\frac{V_0(P)}{4P^2}
\label{eq:Gamma}
\end{align}
for any price $P > 0$.
Notably, by this construction, it immediately follows that the gamma $\frac{\partial^2}{\partial P^2} V_0(P) < 0$ is strictly negative.

\paragraph{Vega:} Finally, consider the sensitivity of the value of the liquidity token to the realized volatility, i.e., the vega of the liquidity token. Due to the dependence of $\gamstr$ on the volatility $\sigma$, the vega has a more complex dependency: 
\begin{align}
\!\!\!\!\!\!\!\!\! \frac{\partial}{\partial \sigma} V_0(P) = \frac{\hat\gamma \sqrt{P} e^{-\frac{1}{2}(r+\frac{\sigma^2}{4})\dt}}{1 - e^{-\frac{1}{2}(r+\frac{\sigma^2}{4})\dt}} \left[\sqrt{\frac{\dt}{2\pi}}e^{-\frac{r^2\dt}{2\sigma^2}} - \frac{\sigma\dt}{4}\frac{\Phi(\frac{(r+\frac{\sigma^2}{2})\sqrt{\dt}}{\sigma})-e^{-r\dt}\Phi(\frac{(r-\frac{\sigma^2}{2})\sqrt{\dt}}{\sigma})}{1 - e^{-\frac{1}{2}(r+\frac{\sigma^2}{4})\dt}}\right]
\label{eq:Vega}
\end{align}
for any price $P>0$.
In contrast to the delta and gamma of this position, the vega does not have a constant sign. In particular, $\frac{\partial}{\partial\sigma} V_0(P) > 0$ for $\sigma > 0$ sufficiently small and $\frac{\partial}{\partial\sigma} V_0(P) < 0$ for $\sigma > 0$ sufficiently large. Therefore, in contrast to the typical derivatives contracts, e.g., a European call option, the liquidity token has a complex dependency on volatility rather than simply being a long position in volatility.
We demonstrate an example of this complex dependency of vega on the volatility in Figure~\ref{fig:vega}. 
Intuitively, when $\sigma > 0$ is small, a tiny increase in volatility will lead to more trading and, therefore, fees collected by the LPs, i.e., a positive vega. On the other hand, when $\sigma>0$ is already very large, a further increase in volatility will, also, increase the probability of the event that the price collapses to zero; this leads to a drop in the price of the token $P$ which results in a negative vega.
\begin{figure}[h!]
\centering
\begin{subfigure}[t]{0.45\textwidth}
\centering
\includegraphics[width=\textwidth]{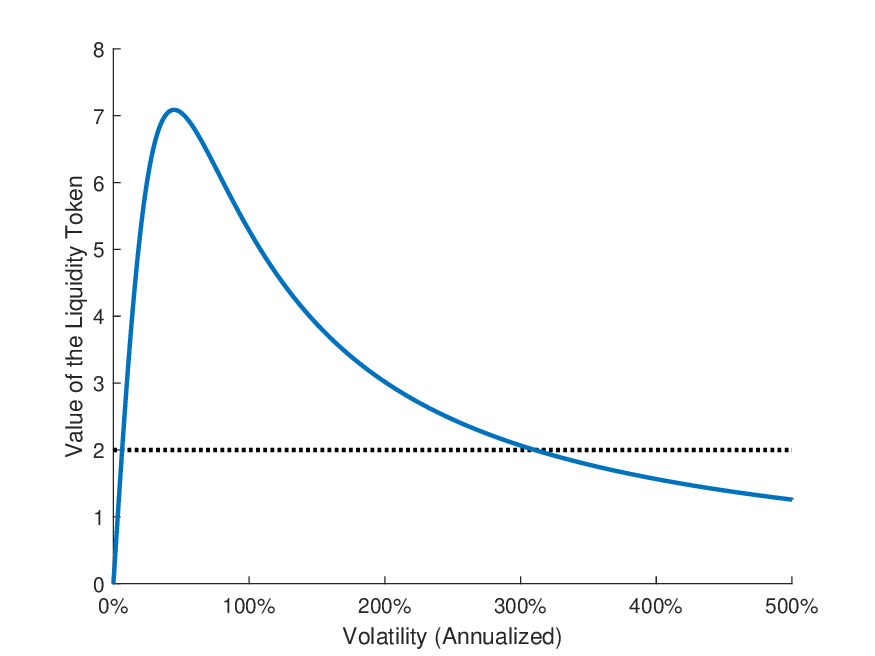}
\caption{Valuation $V_0(P_0)$ of a single liquidity token as a function of (annualized) volatility.}
\label{fig:valuation}
\end{subfigure}
~
\begin{subfigure}[t]{0.45\textwidth}
\includegraphics[width=\textwidth]{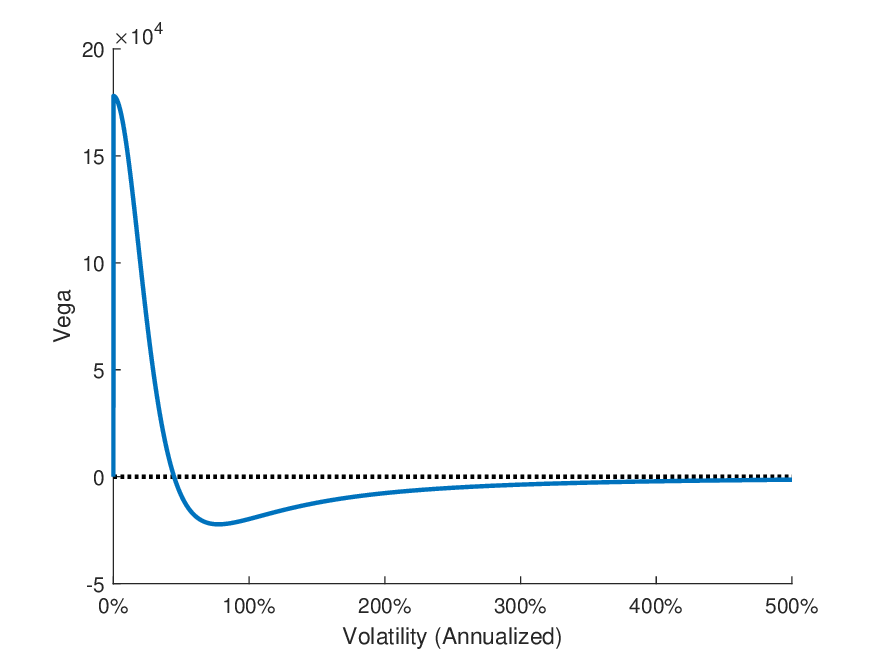}
\caption{Vega $\frac{\partial}{\partial \sigma} V_0(P_0)$ of a single liquidity token as a function of (annualized) volatility.}
\label{fig:vega}
\end{subfigure}
\caption{Valuation and vega with $P_0 = 1$, $\gamma = 5bps$ (i.e., $\hat\gamma = 5.0025bps$), $r = 5\%$ (annualized) and $\dt = 2$ seconds.}
\label{fig:V-vega}
\end{figure}

\section{The Implied Volatility and Estimating the Arbitrage-Free Price of a Liquidity Token}\label{sec:repricing}

Recall from Section~\ref{sec:motivating-example}, the current prevailing market price for a liquidity token within a CPMM is $2\sqrt{P}$. In Section~\ref{sec:cpmm-value}, we found the risk-neutral valuation of these tokens. Within Section~\ref{sec:vol}, we investigate the market implied volatility so that the risk-neutral price coincides with $2\sqrt{P}$. 
However, following Example~\ref{ex:motivating}, Uniswap data readily provides for arbitrage opportunities. Therefore, the fair price of the liquidity token is \emph{not} the market price provided by the CPMM (i.e., is not $2\sqrt{P}$). Within Section~\ref{sec:repricing-vol} we wish to use the above theory on pricing the CPMM in order to calibrate the arbitrage-free pricing of the liquidity token during the period of study. In doing so, we provide a procedure to estimate a new volatility which is ``implied'' by observed data and which can then be used to re-price the liquidity token. 
We conclude by revisiting Example~\ref{ex:motivating} to demonstrate the efficacy of our calibrated volatility for the CPMM liquidity token by investigating the degree to which arbitrage opportunities can be eliminated during the period of study.

\subsection{Implied Volatility}\label{sec:vol}

Assume that the current time ($t = 0$) is a block time. Recall from Theorem~\ref{thm:v2} that the value of a liquidity token is given by $V_0(P_0)$. Further, as discussed in Section~\ref{sec:motivating-example}, the current market price of a liquidity token is $2\sqrt{P_0}$. Therefore, in order to determine the implied volatility, we seek to find $\sigma > 0$ so that $V_0(P_0) = 2\sqrt{P_0}$. In particular, we are seeking the volatility so that the risk-neutral investor is indifferent between investing in the liquidity token or holding the original cash position. Notably, as expressed in the proof of Theorem~\ref{thm:v2}, this differs from the condition that the risk-neutral investor would choose to withdraw the liquidity from the CPMM immediately after depositing.  

\begin{definition}\label{defn:impliedvol}
A volatility $\sigma > 0$ is called an \textbf{implied volatility} if $V_0(P_0) = 2\sqrt{P_0}$ and, with a slight abuse of notation to make the dependence on volatility explicit, $\gamstr(\sigma) \leq \hat\gamma$.
\end{definition}

In the following lemma, we study the implied volatility under any market scenario. Notably, we find three possible situations: (i) no implied volatility exists; (ii) a unique implied volatility exists; or (iii) exactly two implied volatilities exist. Note that this result is consistent with our prior discussion of vega (see, e.g., Figure~\ref{fig:vega}) in which vega is positive for small volatilities and negative for large volatilities. 

\begin{lemma}\label{lemma:v2-sigma}
The implied volatility $\sigstr > 0$ is any volatility such that $\gamstr(\sigstr) = \hat\gamma$.
\begin{itemize}
\item If $r = 0$ then there exists a \emph{unique} implied volatility $\sigstr > 0$.
\item If $r > 0$ then:
    \begin{itemize}
    \item If $\dt > \overline\dt := \sqrt{\frac{8}{\pi}}\frac{\hat\gamma}{(2+\hat\gamma)r}e^{-1/2}$ then \emph{no} implied volatility exists.
    \item If $\dt \leq \overline\dt$ then define $\bar\sigma := r\sqrt{\frac{\dt}{-W\left(-\frac{\pi}{2}\left[\frac{(2+\hat\gamma)r\dt}{2\hat\gamma}\right]^2\right)}}$ and:
        \begin{itemize}
        \item there does \emph{not} exist an implied volatility if $\hat\gamma <\gamstr(\bar\sigma)$;
        \item there exists a \emph{unique} implied volatility $\sigstr = \bar\sigma$ if $\hat\gamma =\gamstr(\bar\sigma)$; and
        \item there exists exactly \emph{two distinct} implied volatilities $\sigstr_1 < \bar\sigma < \sigstr_2$ if $\hat\gamma >\gamstr(\bar\sigma)$.
        \end{itemize}
    \end{itemize}
\end{itemize}
\end{lemma}
\begin{proof}
First, following Definition~\ref{defn:impliedvol} and Theorem~\ref{thm:v2}, assume $\hat\gamma \geq\gamstr(\sigma)$, then $V_0(P_0) = \frac{2\hat\gamma\sqrt{P_0}}{\gamstr(\sigma)}$.  Immediately, since the initial cost of this investment is $2\sqrt{P_0}$, the implied volatility must be such that $\gamstr(\sigma) = \hat\gamma$.

Define $G: \R_{++} \to \R$ such that
\[G(\sigma) := (2+\hat\gamma)\left[1 - e^{-\frac{1}{2}(r + \frac{\sigma^2}{4})\dt}\right] - \hat\gamma\left[\Phi\left(\frac{(r+\frac{\sigma^2}{2})\sqrt{\dt}}{\sigma}\right) - e^{-r\dt}\Phi\left(\frac{(r-\frac{\sigma^2}{2})\sqrt{\dt}}{\sigma}\right)\right]\]
for any $\sigma > 0$. By construction, $\gamstr(\sigma) \geq \hat\gamma$ ($\leq$) if, and only if, $G(\sigma) \leq 0$ (resp.\ $\geq$). Therefore, we can determine the existence properties of the implied volatility by studying the roots of $G$.
As it will be needed later, we note that $G$ is differentiable with derivative
\[G'(\sigma) = e^{-\frac{1}{2}(r+\frac{\sigma^2}{4})\dt}\left[(2+\hat\gamma)\frac{\sigma\dt}{4} - \hat\gamma\sqrt{\frac{\dt}{2\pi}}e^{-\frac{r^2\dt}{2\sigma^2}}\right].\]

First, consider the case with zero risk-free rate $r = 0$. We note that $\lim_{\sigma \searrow 0} G(\sigma) = 0$ (with $\lim_{\sigma \searrow 0} G'(\sigma) = -\hat\gamma\sqrt{\frac{\dt}{2\pi}}$) and $\lim_{\sigma \nearrow \infty} G(\sigma) = 2$. 
Therefore, there exists at least one implied volatility $\sigstr > 0$.
Additionally, $\bar\sigma := \frac{\hat\gamma}{2+\hat\gamma}\sqrt{\frac{8}{\pi\dt}}$ is the unique positive root of $G'$. Therefore, there cannot exist a positive implied volatility below $\bar\sigma$ and there must exist a unique implied volatility $\sigstr > \bar\sigma$.

Assume, now, a strictly positive risk-free rate $r > 0$.
First, we note that $\lim_{\sigma \searrow 0} G(\sigma) = 2(1 - e^{-\frac{1}{2}r\dt}) - \hat\gamma(e^{-\frac{1}{2}r\dt} - e^{-r\dt}) > 0$ and $\lim_{\sigma \nearrow \infty} G(\sigma) = 2$.  Therefore, there exists an implied volatility $\sigstr > 0$ if, and only if, $\inf_{\sigma > 0} G(\sigma) \leq 0$.
To investigate this infimum, we first consider the sign of $G'$. In particular, $G'(\sigma) \leq 0$ if, and only if, $(2+\hat\gamma)\frac{\sigma\dt}{4} - \hat\gamma\sqrt{\frac{\dt}{2\pi}}e^{-\frac{r^2\dt}{2\sigma^2}} \leq 0$ or, equivalently via rearranging terms, $\sigma e^{\frac{r^2\dt}{2\sigma^2}} \leq \frac{4\hat\gamma}{\sqrt{2\pi\dt}(2+\hat\gamma)}$. Noting that $\inf_{\sigma > 0} \sigma e^{\frac{r^2\dt}{2\sigma^2}} = r\sqrt{\dt}e^{1/2}$, $G'(\sigma) > 0$ for every $\sigma > 0$ if, and only if, $r\sqrt{\dt}e^{1/2} > \frac{4\hat\gamma}{\sqrt{2\pi\dt}(2+\hat\gamma)}$, i.e., $\dt > \overline\dt$. Assuming $\dt \leq \overline\dt$, we can find a unique root of $G'$ using the Lambert $W$ function at $\bar\sigma$ as provided in the statement of the lemma. We note that $\dt \leq \overline\dt$ guarantees $-\frac{\pi}{2}\left[\frac{(2+\hat\gamma)r\dt}{2\hat\gamma}\right]^2 \in [-e^{-1},0)$, i.e., $\bar\sigma$ is well-defined and positive.
Assume $\dt \leq \overline\dt$, then since this root $\bar\sigma$ is unique, it must follow that $\inf_{\sigma > 0} G(\sigma) \in \{G(\bar\sigma) , \lim_{\sigma \searrow 0} G(\sigma)\}$.
From this, and recalling the relation between the sign of $G$ and the risk-neutral valuation, the result trivially follows.
\end{proof}

Due to the construction of $\gamstr(\sigma)$, the (set of) implied volatility is independent of the current price $P_0$ and, as noted throughout this work, the total level of liquidity in the CPMM. 
Because all other parameters are fixed by the blockchain ($\dt$), AMM smart contract ($\gamma$), or central bank ($r$), this constancy allows us to consider market structures that admit arbitrage opportunities (i.e., $V_0(P) > 2\sqrt{P}$ for every price $P > 0$) based solely on the realized volatility.
\begin{corollary}\label{cor:arb}
There exist arbitrage opportunities, i.e., $V_0(P) > 2\sqrt{P}$ for every price $P > 0$, if and only if either
\begin{itemize}
\item $r = 0$ and $\sigma < \sigstr$; or
\item $r > 0$, $\dt \leq \overline\dt$, $\hat\gamma > \gamstr(\bar\sigma)$, and $\sigma \in (\sigstr_1,\sigstr_2)$.
\end{itemize}
\end{corollary}
\begin{proof}
This result follows directly from the proof of Lemma~\ref{lemma:v2-sigma}.
\end{proof}
We wish to remind the reader that these arbitrage opportunities are clearly seen in Figure~\ref{fig:valuation}. As Corollary~\ref{cor:arb} specifies, the shape of Figure~\ref{fig:valuation} is general and not specific to the parameters chosen therein.

Before continuing, we want to provide a numerical example to demonstrate all possible outcomes for the set of implied volatilities. We will do this by considering the Polygon blockchain with different fee levels $\gamma$ to demonstrate the different possible settings for the implied volatility. In each of these cases, $\dt < \overline\dt$ by orders of magnitude.
\begin{example}\label{ex:impliedvol}
Consider a CPMM on the Polygon blockchain ($\dt = 2$ seconds). For these examples, recall that the risk-neutral valuation and other formulas employed throughout this work utilize the realized fee level $\hat\gamma = \gamma/(1-\gamma)$. 
For any choice of $\gamma$, there are two possible outcomes. The less interesting possibility is that $V_0(P) < 2\sqrt{P}$ so that depositing liquidity for the price of $2\sqrt{P}$ is expected to lose value, i.e., the expected discounted value of the fees would not cover this initial deposit value.\footnote{Within this work, we have assumed that it is not possible to sell liquidity tokens short and, therefore, no arbitrage opportunity exists if $V_0(P) < 2\sqrt{P}$.}
The more interesting scenario is where $V_0(P) \geq 2\sqrt{P}$ which (with a strict inequality) results in an arbitrage opportunity since the purchase price of the liquidity tokens $2\sqrt{P}$ is below the risk-neutral valuation of the fees. Therefore, in this latter scenario, a rational investor can deposit the liquidity for $2\sqrt{P}$ and, with appropriate hedging, obtain the value $V_0(P)$.\footnote{We refer the interested reader to \cite[Chapter 1.4]{karatzas1998methods} for an initial discussion on arbitrage.}
\begin{figure}[h!]
\centering
\begin{subfigure}[t]{0.45\textwidth}
\includegraphics[width=\textwidth]{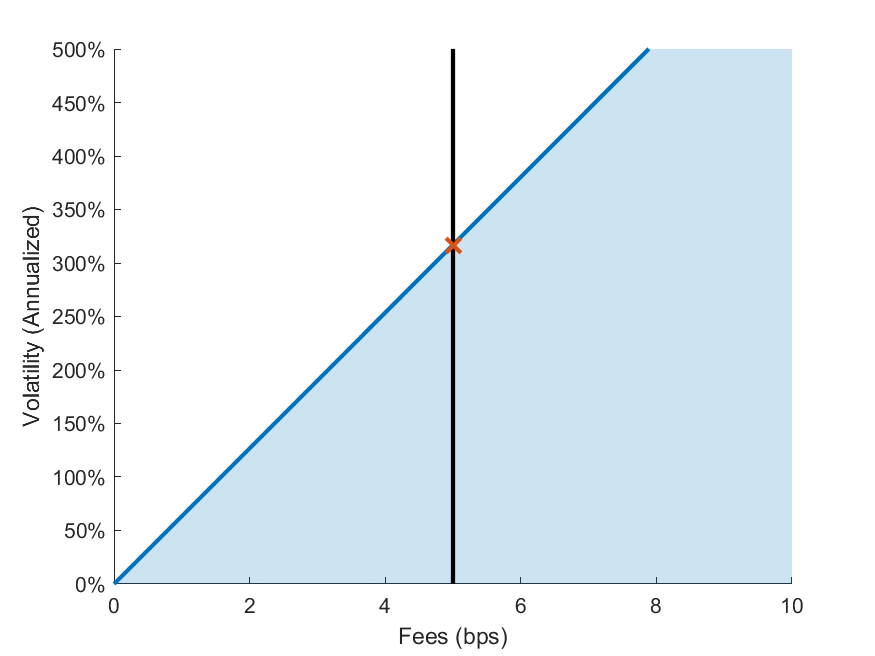}
\caption{$r = 0\%$ with marked point indicating the implied volatility $\sigstr$ at $\gamma = 5bps$.}
\label{fig:impliedvol_r=0}
\end{subfigure}
~
\begin{subfigure}[t]{0.45\textwidth}
\centering
\includegraphics[width=\textwidth]{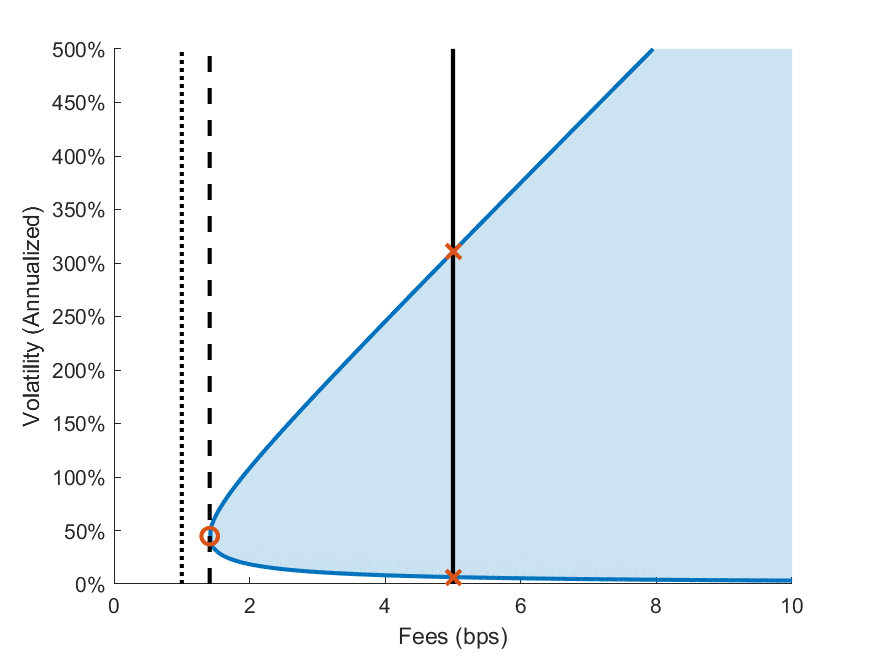}
\caption{$r = 5\%$ (annualized) with marked points indicating the implied volatilities $\sigstr$ at $\gamma \in \{1bps,1.14114bps,5bps\}$.}
\label{fig:impliedvol_r>0}
\end{subfigure}
\caption{Example~\ref{ex:impliedvol}: The shaded region indicates fee-volatility $(\gamma,\sigma)$ pairs that provide arbitrage opportunities.
}
\label{fig:impliedvol}
\end{figure}
\begin{itemize}
\item If $r = 0\%$ and $\gamma = 5bps$ then, by Lemma~\ref{lemma:v2-sigma}, there exists a unique implied volatility $\sigstr \approx 316.75\%$ (annualized). Notably, as proven in Corollary~\ref{cor:arb} and shown in Figure~\ref{fig:impliedvol_r=0}, if $\sigma < \sigstr$ then $V(P_0) > 2\sqrt{P_0}$ and a risk-neutral investor would deposit liquidity at the AMM. In contrast, if $\sigma > \sigstr$, then $V(P_0) < 2\sqrt{P_0}$ and a risk-neutral investor would not pool liquidity at the AMM.
\item If $r = 5\%$ (annualized) and $\gamma = 1bps$ then $\overline\dt \approx 8.48$ hours and $\bar\sigma \approx 0.3168$ with $\gamstr(\bar\sigma) \approx 1.4962bps > \hat\gamma$. Therefore, there does not exist an implied volatility at this fee level, i.e., a risk-neutral investor would never deposit liquidity at the AMM. 
This is shown in Figure~\ref{fig:impliedvol_r>0} by the dotted line.
\item If $r = 5\%$ (annualized) and $\gamma \approx 1.4114bps$ (so that $\hat\gamma \approx 1.4116bps$) then $\overline\dt \approx 11.97$ hours and $\bar\sigma \approx 0.4472$ with $\gamstr(\bar\sigma) = \hat\gamma$. Therefore, the unique implied volatility is given by $\sigstr = \bar\sigma \approx 44.72\%$ (annualized). 
Notably, if $\sigma \neq\sigstr$ then a risk-neutral investor would not pool liquidity at the AMM.
\item If $r = 5\%$ (annualized) and $\gamma = 5bps$ then $\overline\dt \approx 42.40$ hours and $\bar\sigma \approx 1.5846$ with $\gamstr(\bar\sigma) \approx 2.7002bps < \hat\gamma$. Therefore, there exists two implied volatilities:
$\sigstr_1 \approx 6.44 < \bar\sigma < \sigstr_2 \approx 310.47\%$ (annualized).  
Notably, as proven in Corollary~\ref{cor:arb} and shown Figure~\ref{fig:impliedvol_r>0}, if $\sigma \in (\sigstr_1,\sigstr_2)$ then $V(P_0) > 2\sqrt{P_0}$ and a risk-neutral investor would deposit liquidity at the AMM. In contrast, if $\sigma \in (0,\sigstr_1) \cup (\sigstr_2,\infty)$ then $V(P_0) < 2\sqrt{P_0}$ and a risk-neutral investor would not pool liquidity at the AMM.
\end{itemize}
\end{example}

\begin{remark}
As evidenced in Example~\ref{ex:impliedvol} above (and comparing to the parameters of, e.g., Example~\ref{ex:motivating}), non-uniqueness of the implied volatility can easily occur in practice when $r > 0$. It becomes important to understand which of $\{\sigstr_1,\sigstr_2\}$ should be quoted. Herein, as $\sigstr_2$ converges to the \emph{unique} solution $\sigstr$ when the risk-free rate approaches 0, i.e., $\lim_{r \searrow 0} \sigstr_2(r) = \sigstr(0)$ with dependence on the risk-free rate made explicit, we take this upper implied volatility to be the more meaningful setting. In comparison, the lower implied volatility $\sigstr_1$ converges to 0 as the risk-free rate approaches 0, i.e., $\lim_{r \searrow 0} \sigstr_1(r) = 0$.
\end{remark}

\subsection{Estimating the Arbitrage-Free Price of a Liquidity Token}\label{sec:repricing-vol}

In contrast to the market price $2\sqrt{P}$ utilized for the implied volatility in Section~\ref{sec:vol} above, herein we want to calibrate the pricing of the liquidity token to the observed data.
To do this, we will first find the volatility $\sigM > 0$ so that we observe martingale pricing in the data. 
Specifically, over a given historical time period (e.g., 1 month), we find the volatility $\sigM > 0$ so that the delta hedged position under pricing with $\gamstr(\sigM)$ at the end of the period has the same value as the start of the period. We note that this historically calibrated volatility differs from the more traditional historical volatility based on repeated snapshots of the data. To demonstrate the long-term stability of this method, we revisit Example~\ref{ex:motivating} to calibrate the volatility to empirical data. Within this example, we find that the proposed method can accurately re-price the liquidity token so as to effectively eliminate the arbitrage opportunities encountered in practice based on the prevailing market price of $2\sqrt{P}$ even months after the calibration period.
\begin{example}\label{ex:repricing}
Consider the USDC/WETH Uniswap v3 pool on the Polygon blockchain considered in Example~\ref{ex:motivating} (i.e., $\dt = 2$ seconds and $\gamma = 5bps$, i.e., $\hat\gamma \approx 5.0025bps$). As in Example~\ref{ex:motivating}, throughout this discussion we will set $r = 0\%$. 
Recall, also, from the first case in Example~\ref{ex:impliedvol} there exist a unique implied volatilities for the market price of $2\sqrt{P}$ given by $\sigstr \approx 316.75\%$ (annualized). 
Specifically, calibrating the volatility $\sigM > 0$ to January 2023, we find $\sigM \approx 143.75\%$ (annualized) which we then apply throughout the full 2023 calendar year. We note that, 
with this approximation of the calibrated volatility, we are able to re-price the liquidity token by noting that $\hat\gamma/\gamstr(\sigM) \approx 2.2048 > 1$, i.e., Uniswap is underpricing the liquidity token by a factor of $2.2048$.

In Figure~\ref{fig:repricing-hedge}, we compare the delta hedged position under this re-pricing (solid black line) compared to that of the original market price (dashed blue line); for direct comparisons, we assume 1 USDC was invested in the pool for both cases (yielding different number of liquidity tokens). Notably, the hedging error under repricing is an order of magnitude lower and fluctuates around the initial price of 1 compared to the market pricing (as observed already in Example~\ref{ex:motivating}). This minimal hedging error demonstrates that this updated pricing can be viewed as an arbitrage-free price of the liquidity token. Despite the nearly perfect hedge constructed for the re-priced CPMM, we note that there exist a select few times when its value jumps. These jumps correspond to times at which the WETH price experiences sudden, large movements which are not delta hedgeable, see Figure~\ref{fig:motivating-full}. As exected, the power of this repricing decreases as we move further into our out-of-sample period at the end of 2023 with a noticeable updward trend from October through December; recall we are using a constant calibrated volatility $\sigM$ based on only January's data rather than allowing it to fluctuate throughout the period. Alternative AMM constructions which permit variable implied volatilities are discussed in Section~\ref{sec:discussion}.
\begin{figure}[h!]
\centering
\includegraphics[width=0.6\textwidth]{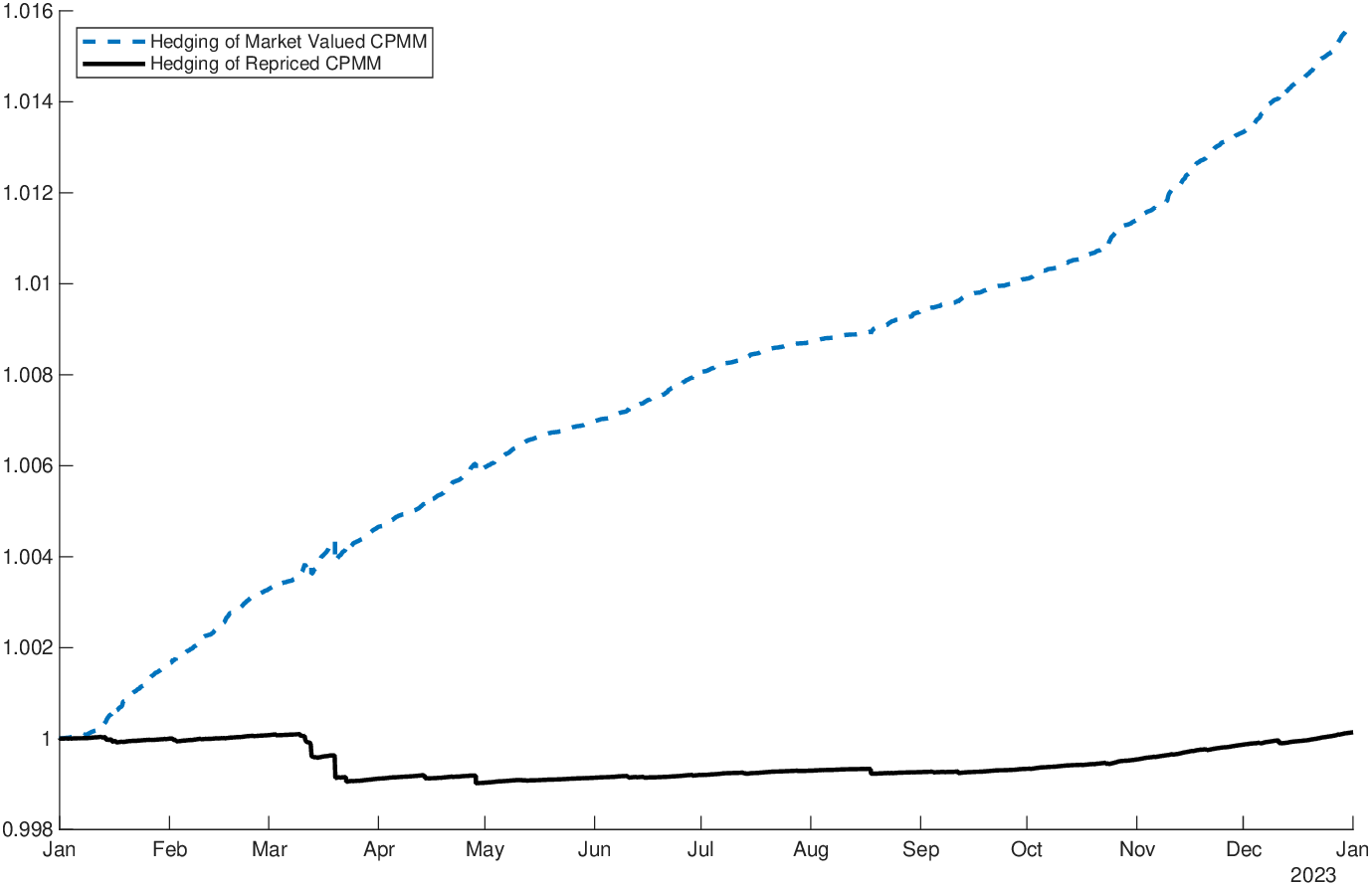}
\caption{Example~\ref{ex:repricing}: Comparison of the delta hedged position of 1 USDC investment in a liquidity token under market pricing (blue dashed line) and with the calibrated arbitrage-free price (black solid line). Volatility is calibrated on January 2023 and applied throughout calendar year 2023.}
\label{fig:repricing-hedge}
\end{figure}
\end{example}

\begin{remark}
Herein we have focused on hedging the liquidity token itself rather than the impermanent loss as in, e.g.,~\cite{fukasawa2023weighted,lipton2024unified}. As the impermanent loss is just the difference between the value of the liquidity token (including any collected fees) and the value of the initial position $P_t x_0 + e^{rt} y_0 = P_t/\sqrt{P_0} + e^{rt} \sqrt{P_0}$ which was used to mint the liquidity token.\footnote{We wish to note that this construction may need adjustments if the price of the liquidity token grows beyond $2\sqrt{P_0}$ at time $0$ as we propose in Example~\ref{ex:repricing} above.} As the initial position is static, hedging the impermanent loss reduces to appropriately hedging the liquidity token itself and, therefore, our results easily generalize to this more widely studied problem.
\end{remark}

\section{Discussion}\label{sec:discussion}
Given that we found that the prevailing market prices for CPMMs can exhibit arbitrage opportunities, it is important to construct new AMM designs that permit freely floating pricing for liquidity tokens. That is, where the price of a liquidity token depends explicitly on the number of outstanding tokens which does not exist at present. This would also permit, e.g., varying implied volatilities instead of the flat implied volatility based on the current pricing scheme (as given in Section~\ref{sec:vol}). 

Notably, if a secondary market were created for liquidity tokens, so long as the CPMM quotes a price of $2\sqrt{P}$ then, through a no-arbitrage argument, the prices would never vary from that level. To accomplish a meaningful, freely floating price, the CPMM would either need to vary the price of liquidity tokens internally or vary the fee rate being charged to swappers. In either case, so that no external oracles are required by the CPMM smart contract, these constructions need to be dependent on the number of outstanding liquidity tokens $L > 0$. Below we explore these two novel pricing frameworks.

\paragraph{Variable minting/burning costs $V(P,L)$:} 
Following the mispricing identified in Example~\ref{ex:repricing}, introducing variable minting or burning values of the next (marginal) liquidity token $V(P,L) = 2v(L)\sqrt{P}$ for some strictly increasing function $v: \R_{++} \to \R_{++}$ of the outstanding liquidity tokens $L > 0$. 
To recreate such a structure, the CPMM structure needs to be updated. Specifically, taking $xy = \ell(L)$ -- generalizing $xy = L^2$ as taken for CPMMs at present -- for $\ell: \R_{++} \to \R_{++}$ twice continuously differentiable with $\ell'(L) > 0$ and $2\ell(L)\ell''(L) > \ell'(L)^2$ for every $L > 0$ results in the updated pricing $v(L) = \frac{\ell'(L)}{2\sqrt{\ell(L)}}$. Note that this updated pricing, generally, depends explicitly on the outstanding liquidity tokens $L$.
For example, inspired by the classical quadratic approach to liquidity tokens, taking $\ell(L) = L^{2\alpha}$ for $\alpha > 1$ results in the updated pricing scheme with $v(L) = \alpha L^{\alpha-1}$. 
With this updated pricing, the price of liquidity tokens will fluctuate as the total amount of market liquidity changes. Akin to traditional derivatives markets, LPs can use the implied volatility (so that $V_0(P;\sigma) = 2v(L)\sqrt{P}$ to generalize the discussion of Section~\ref{sec:vol}) in order to evaluate the performance of their investment.
However, since other CPMMs exist in the market, a savvy investor would never pay more than $2\sqrt{P}$ for this contract. Therefore, though in theory having fully variable pricing is possible, it would never succeed unless, and until, all CPMMs allow for variable pricing.

\paragraph{Variable fee rate $\gamma(L)$:} 
To overcome the aforementioned issue in which the $2\sqrt{P}$ pricing at other CPMMs (such as at Uniswap pools) limits the liquidity available under variable minting costs, we can consider variable fee rates instead.
That is, consider a CPMM in which the fees charged are provided by the strictly decreasing mapping $\gamma: \R_{++} \to \R_{++}$ of the outstanding liquidity tokens $L > 0$. 
Now, instead of updating the market price $2\sqrt{P}$, the CPMM has an updated risk-neutral valuation $V_0(P,L) = 2\hat\gamma(L)\sqrt{P}/\gamstr$ per liquidity token as per Theorem~\ref{thm:v2}. Recall that $\gamstr$, as defined in \eqref{eq:gamstr}, is constant for fixed $r,\sigma$ and is independent of the liquidity $L$. With this construction, the ratio $\hat\gamma(L)/\gamstr$ provides a delineation of whether LPing is a good investment or not, assuming all future fees remain constant at $\gamma(L)$. As opposed to the constant fee setting, as investors may have different beliefs about volatility $\sigma$, they each estimate $\gamstr$ differently leading to equilibrium liquidity $L^*$; in this way, if an investor believe that $\hat\gamma(L^*) < \gamstr$ then they will not invest while an investor that believes that $\hat\gamma(L^*) > \gamstr$ will choose to increase the CPMM liquidity. 
We recommend that $\gamma_0 := \lim_{L \searrow 0} \gamma(L)$ is set sufficiently small (e.g., $30bps$ as was used in all Uniswap v2 pools) so as to encourage LPing 
and $\lim_{L \nearrow \infty} \gamma(L) = 0$ so that any reasonable price level can be supported. For example, $\gamma(L) = \gamma_0 \exp(-\alpha L)$ or $\gamma(L) = \gamma_0 (1 + \alpha L)^{-1}$ for parameter $\alpha > 0$ provides a control over the fee dependence on liquidity.
Following the logic of Section~\ref{sec:vol}, an implied volatility can be given that now has explicit dependence on the level of liquidity $L$ in the CPMM pool. 
The primary drawback to this construction is that the LP cannot know the realized fee rate for their investment when making the purchase; instead she will need to continuously review this investment to determine if it remains advantageous.

\section{Conclusion}\label{sec:conclusion} 
Evaluating data from a CPMM indicates that {hedging the liquidity token from, e.g., Uniswap, can result in arbitrage opportunities}. 
The constant pricing scheme considered in practice can be viewed akin to the pre-Black-Scholes world for derivatives. Within this work we have determined a risk-neutral pricing theory for CPMMs. With this theory we have revisited the data to determine an approximating arbitrage-free price. Furthermore, we propose two novel AMM designs so that the pricing of the liquidity tokens are variable in time based on the demand for such tokens.


Though we focused solely on the CPMM within this work, we conjecture similar mispricing  can be found in other AMM designs. In particular, we wish to highlight the concentrated liquidity designs of Uniswap v3; we conjecture that the optimal investment strategy for the concentrated liquidity would include finite stopping times which may complicate the risk-neutral pricing. In comparison to the flat implied volatility curve for the CPMM (see Section~\ref{sec:vol}), the concentrated liquidity structure would permit a DeFi implied volatility curve; such a structure could be of great interest for practitioners to more accurately price and hedge DeFi risks. We believe a study of such constructions would be of great interest.

Finally, in Section~\ref{sec:discussion}, we proposed two frameworks for variable pricing of liquidity tokens in a CPMM. As far as the authors are aware, no AMMs have implemented liquidity-adjusted pricing or fees. As such, there is no data in which to validate the performance of either proposed approach in practice. These constructions require further study and, potential, implementation. If implemented, an instantaneous volatility index can be plotted over time which could provide interesting insights for sophisticated investors; this is in contrast to the historical volatility or calibrated volatility (as presented in Section~\ref{sec:repricing-vol}) which are, inherently, backwards looking measures. 

\bibliographystyle{plain}
\small{\bibliography{bibtex2}}

\end{document}